\documentclass[runningheads]{llncs}
\usepackage[utf8]{inputenc}
\usepackage{graphicx}
\usepackage{algorithm}
\usepackage{algpseudocode}
\usepackage{amssymb,amsmath,amsfonts}
\usepackage{xspace,xpunctuate}
\usepackage{subfig}

\def\Nesetril{Ne\v{s}et\v{r}il\xspace}

\def\Erdos{Erd\H{o}s\xspace}
\def\Renyi{R\'{e}nyi\xspace}
\def\Barabasi{Barabási\xspace}

\def\ErRe{\Erdos--\Renyi}

\def\ErReGrs{\ErRe graphs\xspace}
\def\BaAl{\Barabasi--Albert\xspace}

\def\BaAlGrs{\BaAl graphs\xspace}

\DeclareMathOperator{\npclass}{\mathsf{NP}}

\DeclareMathOperator{\wclass}{\mathsf{W[1]}}

\def\gnd{G(n,d/n)}
\def\gndn{G(n,d(n)/n)}

\def\P{\mathbb P}
\def\E{\mathbb E}

\def\aas{a.a.s\xperiod}

\newcommand{\topnab}{\mathop{\widetilde \triangledown}}
\def\topgrad_#1{\widetilde \nabla\!_{#1}}

\newcommand*{\eg}{e.g\xperiod}
\newcommand*{\ie}{i.e\xperiod}

\makeatletter
\newcommand*{\etc}{%
    \@ifnextchar{.}%
        {etc}%
        {etc.\@\xspace}%
}
\makeatother

\begin{document}
%
%
\pagestyle{headings}  
%

%
%
\title{Local Structure Theorems for \ErRe Graphs and their Algorithmic
Applications}
\titlerunning{Local Structure Theorems for \ErRe Graphs}
%
\author{Jan Dreier\inst{1} \and Philipp Kuinke\inst{1} \and
Ba Le Xuan\inst{2} \and Peter Rossmanith\inst{1}}
%
%
\tocauthor{Jan Dreier, Philipp Kuinke, Ba Le Xuan and Peter Rossmanith}
\institute{ Theoretical Computer Science, Dept. of Computer Science,\\RWTH Aachen University, Aachen, Germany\\
\texttt{\{dreier,kuinke,rossmani\}@cs.rwth-aachen.de}
\and
The Sirindhorn International Thai-German Graduate School of
Engineering,\\
King Mongkut’s University of Technology North Bangkok, Thailand\\
\texttt{ba.l-sse2015@tggs-bangkok.org}}

\maketitle              

\begin{abstract}
We analyze local properties of sparse \ErRe graphs, where
$d(n)/n$ is the edge probability. In particular we study the behavior
of very short paths. For $d(n)=n^{o(1)}$ we show that $\gndn$ has
asymptotically almost surely (\aas) bounded local treewidth and
therefore is \aas nowhere dense. We also discover a new and simpler
proof that $\gnd$ has \aas bounded expansion for constant~$d$. The
local structure of sparse \ErReGrs is very special: The
$r$-neighborhood of a vertex is a tree with some additional edges,
where the probability that there are $m$ additional edges decreases
with~$m$. This implies efficient algorithms for subgraph isomorphism,
in particular for finding subgraphs with small diameter. Finally,
experiments suggest that preferential attachment graphs
might have similar properties after deleting a small number of
vertices.

\end{abstract}
\section{Introduction}

One of the earliest and most intensively studied random graph models is the
\ErRe model~\cite{bollobas_2001,erdos}. Graphs from this class are usually
depicted as a random variable $G(n,p)$, which is a graph consisting of $n$
vertices where each pair of vertices is connected independently uniformly at
random with probability $p$. The edge probability $p$ may also depend on the
size of the graph, \eg, $p=d/n$. Many properties of \ErReGrs are well studied
including but not limited to, threshold phenomena, the sizes of components,
diameter, and lengths of paths~\cite{bollobas_2001}. One particular impressive
result is the $0$-$1$ law: Let $\varphi$ be a first-order formula.  If we take
a random graph $G=G(n,1/2)$, then the probability of $G\models\varphi$ is
either $0$ or $1$ as $n\to\infty$~\cite{fagin1976}.

``Instead of the worst case running time, it is also interesting to consider
the average case. Here even the most basic questions are wide open.'' as Grohe
puts it~\cite{Grohe07logic}. One can find an optimal coloring of $G(n,p)$ in
expected linear time for $p<1.01/n$~\cite{CT2003}. The $0$-$1$ law on the
other hand has (not yet) an efficient accompanying algorithm that can decide
whether $G\models\varphi$ for $G=G(n,1/2)$ and a fixed formula~$\varphi$.

One possibility to open up a whole graph class to efficient algorithms
are algorithmic meta-theorems.  Such meta-theorems were developed for more and
more general graph classes: planar, bounded genus, bounded degree, $H$-minor
free, $H$-topological minor free \etc.  In all these graph classes we can
decide properties that are expressible in first-order logic in linear time for
a fixed formula $\varphi$~\cite{4276571,Flum:2002}. Unfortunately, random
graph classes do not belong to any of these classes. For example $G(n,1.1/n)$
has \aas linear treewidth and does contain constant-size cliques of arbitrary
size~\cite{Gao:2012}. Recently, however, graph classes of bounded expansion
were introduced by \Nesetril and Ossona~de~Mendez~\cite{NOdM08}.  These
classes also admit linear time FO-model checking and generalize the older
meta-theorems~\cite{DKT2013}. The most general model checking algorithm runs
in time $O(n^{1+\epsilon})$
on nowhere-dense classes~\cite{Grohe:2014}. In $\gnd$,
the value $d$ is the expected density of a random graph. For constant $d$ it
was shown that $\gnd$ has \aas bounded expansion~\cite{NESETRIL2012350}.
Unfortunately, this does not automatically imply that one can test first-order
properties on $\gnd$ in linear (expected) time, but only that we can test such
a property in linear time with a failure probability of~$o(1)$ while the
expected runtime might be unbounded. This is for example the case if the
runtime grows faster than the failure probability converges to zero. One
example of an (expected-time) fpt-algorithm is one that finds a $k$-clique in
$G(n,p(n))$ in time $f(k)n^{O(1)}$, for many choices of
$p$~\cite{FOUNTOULAKIS201518}.

In Section~\ref{sec:erdos} we find an easier proof for the fact that
$\gnd$ has \aas bounded expansion for constant $d$ and give concrete
probability bounds, which were missing up to now.   Then we investigate
local properties of \ErReGrs. The expected density of $G(n,d(n)/n)$ is $d(n)$
and therefore, if $d(n)$ is not constant, unbounded. This implies that
$G(n,d(n)/n)$ does \aas not have bounded expansion. Nevertheless, we show that
subgraphs with small diameter are tree-like with only a few additional edges.
From this it follows that $G(n,n^{o(1)}/n)$ has \aas locally bounded treewidth,
which implies that they are \aas nowhere dense. Locally bounded
treewidth~\cite{locallyboundedtreewidth} and more generally, locally excluding
a minor~\cite{DBLP:conf/lics/DawarGK07} are useful concepts for developing
first-order model checking algorithms that run in time $O(n^{1+\epsilon})$.

We discussed that a random graph class that is \aas nowhere dense or
has \aas bounded expansion may not directly admit efficient algorithms.
It is known that one can check first-order properties in
$\gndn$ in time $O(g(|\varphi|)n^{1+o(1)})$ for $d(n) = n^{o(1)}$ and some
function $g$~\cite{Grohe07logic,Grohe2001}. For constant $d$ one can check
first-order properties in time $O(g(|\varphi|)n)$. In
Section~\ref{sec:algorithms} we use the locally tree-like structure of
\ErReGrs to construct an efficient algorithm for subgraph isomorphism. We show
that one can find a subgraph $H$ with $h$ vertices and radius $r$ in $\gndn$
in time $2^{O(h)}(d(n)\log n)^{O(r)}n$, while a naive algorithm may need time
$O(d(n)^h n)$. Therefore, our method may be faster for finding large pattern
graphs with small radius.

It can be argued that \ErReGrs are not a good model for real-world networks
and therefore efficient algorithms for \ErReGrs admit only limited practical
applications. Recently, there were more and more efforts to model real world
networks with random graph models. One candidate to meet this goal were the
\BaAlGrs, which use a preferential attachment paradigm to produce graphs with
a degree distribution that tries to mimic the heavy-tailed distribution
observed in many real-world networks~\cite{Barabasi509}.

This model is particularly interesting from the point of mathematical analysis
because of its simple formulation and interesting characteristics, which is
why they have been widely studied in the literature~\cite{cohen,kamrul,klemm}.
It was also shown that this model does \emph{not} have \aas bounded
expansion~\cite{StrucSpars}.

In Section~\ref{sec:experiments} we discuss experiments to see how similar the
local structure of \BaAlGrs is to \ErReGrs.  Not surprisingly, it seems that
they are quite different and the former 
contain dense subgraphs and are likely to be
somewhere dense. If we, however, remove the relatively small dense early part
of these graphs, the local structure of the remaining part looks quite similar
to \ErReGrs and indicators hint that the remaining part is indeed nowhere
dense. As the dense part is quite small it gives us hope that hybrid
algorithms exist that combine different methods for the dense part and the
structurally simple part.  To search for a subgraph $H$, for example, could be
done by guessing which vertices of $H$ lie in the dense part and then using
methods from Section~\ref{sec:algorithms} to find the remaining vertices in
the simple part.

\section{Preliminaries}
\label{sec:prelim}
In this work we will denote probabilities by $\P[\ldots]$ and expectation by
$\E[\ldots]$. We use common graph theory notation~\cite{diestel}. For a graph
$G$ let $V(G)$ be its vertex set and $E(G)$ its edge set. For $v \in V(G)$ we
denote the $r$-neighborhood of $v$ by $N_r(v)$. The degree of a vertex $v$ in
graph $G$ is denoted by $\text{deg}(v)$. We write $G^\prime \subseteq G$ if
$G^\prime$ is a subgraph of $G$. For $X \subseteq V(G)$ we denote by $G[X]$
the subgraph of $G$ that is induced by the vertices in~$X$. The graph
$G[V(G)-X]$ obtained from $G$ by deleting the vertices in $X$ and their
incident edges, is denoted by $G-X$. The treewidth $\text{tw}(G)$ of a graph
is a measure how tree-like a graph is. We denote \ErRe graphs by a random
variable $\gnd$ and distinguish between graphs with constant $d$ and graphs
$\gndn$, where we allow $d$ to grow (slowly) with $n$. We will use various
ways to measure the sparsity of a graph or graph class.

\begin{definition}[Shallow topological minor~\cite{NOdM08}]
\label{def:shallowtopminor+}
    A graph $M$ is an \emph{$r$-shallow topological minor} of~$G$ if $M$ is
    isomorphic to a subgraph $G'$ of~$G$ if we allow the edges of $M$ to be
    paths of length up to $2r+1$ in $G'$.
    We call $G'$ a \emph{model of $M$ in $G$}. For simplicity we assume by
    default that $V(M) \subseteq V(G')$ such that the isomorphism between $M$
    and $G'$ is the identity when restricted to $V(M)$. The vertices $V(M)$
    are called \emph{nails}\footnote{also known as principal vertices} and the
    vertices $V(G') \setminus V(M)$ \emph{subdivision vertices}.
    The set of all $r$-shallow topological minors of a graph $G$ is
    denoted by $G \topnab r$.
\end{definition}
With that we can define the clique size over all topological
minors of $G$ as $$\omega(G \topnab r) = \max_{H\in G \topnab r }
\omega(H).$$

\begin{definition}[Topological grad~\cite{sparsity}]
For a graph $G$ and an integer $r \geq 0$, the topological grad at depth $r$
is defined as
$$
\topgrad_r(G) = \max_{H \in G \topnab r} \frac{|E(H)|}{|V(H)|}
$$
For a graph class $\mathcal{G}$, define $\topgrad_r(\mathcal{G}) =
\sup_{G\in\mathcal{G}} \topgrad_r(G)$.
\end{definition}

\begin{definition}[Bounded expansion~\cite{sparsity}]
A graph class $\mathcal{G}$ has bounded expansion if and only if there exists
a function $f$ such that
$\topgrad_r(\mathcal{G}) < f(r)$
for all $r\geq 0$.
\end{definition}

\begin{definition}[Locally bounded treewidth]
A graph class $\mathcal{G}$ has locally bounded treewidth if and only if there
exists a function $f$, such that for all $r\geq 0$ every subgraph with radius
$r$ has treewidth at most $f(r)$.
\end{definition}

\begin{definition}[Nowhere dense~\cite{sparsity}]
  A graph class~$\mathcal{G}$ is nowhere dense if there exists a function~$f$
  such that $\omega( G \topnab r ) < f(r)$
  for all $G\in\mathcal{G}$ and all~$r \geq 0$.
\end{definition}
If a graph class has locally bounded treewidth it is also nowhere
dense~\cite{sparsity}.

\section{Local Structure and Algorithmic Applications}
\label{sec:erdos}
In this section, we observe the local structure of \ErReGrs and how to exploit
it algorithmically. It is already known that \ErReGrs have \aas bounded
expansion if the edge probability is $d/n$ for constant
$d$~\cite{NESETRIL2012350}. We present a simpler proof via a direct method,
that also gives concrete probability bounds. The original proof did not give
such concrete bounds so we feel that this new proof has applications in the
design of efficient algorithms. To make our calculations easier we assume that
$d\geq2$, since \ErReGrs are only sparser for smaller $d$, our techniques will
also work in this case.

\subsection{Bounded Expansion}
\label{sec:bounded_exp}
The technique we use to bound the probability that certain shallow
topological minors exists is to bound the probability that a path of length
at most $r$ exists between two arbitrary vertices.

\begin{lemma}\label{lem:pathr}
Let $p_r$ be the probability that there is a path of length at most $r$
between two arbitrary but fixed vertices in $\gnd$. It holds that
    $$
    \frac{d}{n} \leq p_r \leq \frac{2d^r}{n}.
    $$
\end{lemma}

\begin{proof}
Since all edges are independent, we do not need to identify the start and
end vertices of the path.
We prove by induction over $r$ that the probability of the existence of a path
of length exactly $r$ is bounded by $\frac{d^r}{n}$. For $r=1$ the statement
holds: $p_1 \leq \frac{d}{n}$.  The probability of
a path of length $r$ is at most that of some path of length $r-1$
times the probability of a single edge:
$$
p_r \leq \sum_{k=0}^n p_{r-1} p_1 \leq
\sum_{k=0}^n \frac{d^{r-1}d}{n^2} \leq \frac{d^r}{n}
$$
By using the union bound and assuming that $d\geq 2$, the joint probability is bounded by $\frac{2d^r}{n}$.
\qed
\end{proof}
Having this bound in place, we can show that $\gnd$ has \aas no $r$-shallow
topological minors of large density from which it follows that they are
contained in a graph class of bounded expansion \aas. 

\begin{theorem}
\label{thm:bounded_exp}
$\gnd$ is \aas contained in a graph class of bounded expansion.  In
particular, for $d\geq16$ the probability that such a random graph contains
some $r$-shallow topological minor of size $k$ and at least $8kd^{2r+1}$ edges is
at most $\max\{n^{-2k},2^{-n^{2/3}}\}$. For $d<16$ the same result holds for
at least $8k{16}^{2r+1}$ edges.
\end{theorem}

\begin{proof}
We will now investigate the probability that a random graph $G=\gnd$ contains
some model of an $r$-shallow topological minor $H$ with nail set
$v_1,\ldots,v_k$.  Such a model consists of the nails themselves and vertex-disjoint
paths of maximal length $q=2r+1$ between them.  Each such path models one
edge of~$H$.  Assume that $V(H)=\{u_1,\ldots,u_k\}$ and that $u_i$ is modeled
by $v_i$.  Then an edge $u_iu_j\in E(H)$ is modeled by a path from $u_i$ to
$u_j$ in~$G$ and all these paths are vertex-disjoint.  What is the probability
that such a model exists? A first path exists with a probability of~$p_{q}$.
The probability that the second path exists under the condition that it does
not cross all candidates for the first path is slightly less than~$p_{q}$.
Continuing this argument shows that the probability of finding such a model is
at most $p_{q}^{|E(H)|}$ and more specifically it is at most the probability of
getting $|E(H)|$ heads after $|E(H)|$ independent coin tosses with a head
probability of~$p_{q}$.  Moreover, this implies that the probability of finding
a model for some $H$ with $m$ edges is at most the probability of getting at
least $m$ heads after tossing $k\choose2$ such coins.

Let $X$ be the sum of $k\choose2$ independent Bernoulli variables with
$\Pr[X=1]=p_{q}$. Using the bounds of Lemma~\ref{lem:pathr} we have
$$
\frac{4k^2}{n} \leq \frac{k^2d}{n} \leq \E[X] \leq \binom{k}{2} \frac{2d^{q}}{n} \leq
\frac{k^2d^{q}}{n}.
$$
Using Chernoff bounds with $\delta=\frac{8n}{k}-1$ we get
$$
\P\Bigl[X>(1+\delta)\frac{k^2d^{q}}{n}\Bigr]
\leq\P[X>(1+\delta)~\E[X]]
\leq\Bigl(\frac{e^\delta}{(1+\delta)^{1+\delta}}\Bigr)^{\E[X]}
\leq\Bigl(\frac{ek}{8n}\Bigr)^{32k}.
$$
This means for fixed $k$ nails, with probability of at most
$(\frac{ek}{8n})^{32k}$ the graph $G$ contains a model for an $r$-shallow
topological minor with these nails and at least $(1+\delta)\frac{k^2d^{q}}{n} =
8kd^{q}$ edges. The density of such a topological minor is therefore $8d^{q}$.
There are only ${n\choose k}\leq\bigl(\frac{ne}{k}\bigr)^k$ possibilities to
choose the nails, so an $r$-shallow
topological minor with $k$ nails and density at least $8d^{q}$ exist with a
probability of at most
$\bigl(\frac{ne}{k}\bigr)^k\bigl(\frac{ek}{8n}\bigr)^{32k}$,
which is $n^{-2k}$ if $k\leq n^{2/3}$.  For bigger $k$ it is bounded by
$2^{-n^{2/3}}$. Therefore, every $r$-shallow topological minor in $G$ has \aas
a density of at most $8d^{q}$.
\end{proof}

\subsection{Locally Simple Structure}
\label{sec:local_structure}
It is known that even for constant $d$ the treewidth of $\gnd$
grows with $\Omega(n)$~\cite{Gao:2012}. Furthermore, $G(n,d(n)/n)$
does a.a.s not have bounded expansion if $d(n)$ is unbounded.
We now show that $G(n,n^{o(1)}/n)$
nevertheless has locally bounded treewidth and thus is \aas
nowhere dense. We start by counting the expected number of occurrences of
a certain subgraph in $\gndn$.
%
\begin{lemma}
    \label{lem:num_h}
    The expected number of induced subgraphs with $k$ vertices
    and at least $k+m$ edges in $\gndn$ is at most $k^{2k+2m}d(n)^{k+m}/n^{m}$.
\end{lemma}
\begin{proof}
There are $\binom{n}{k}\leq n^k$ induced subgraphs $H$ of size~$k$ in~$G$.
For each such $H$ there are $\tbinom{\binom{k}{2}}{k+m}\leq
k^{2k+2m}$ ways to choose $k+m$
edges.  The probability that these $k+m$ edges are present in $H$ is
then exactly $(d(n)/n)^{k+m}$ and the probability that $H$ has $k+m$ edges 
is at most $k^{2k+2m}(d(n)/n)^{k+m}$.  Finally, the expected number of such
induced subgraphs is at most $k^{2k+2m}d(n)^{k+m}/n^m$.\qed
\end{proof}
From Lemma~\ref{lem:num_h} we can conclude a well known property of \ErReGrs:
The expected number of cycles of fixed length $r$ is $O(d(n)^r)$ (which is
a constant if $d$ is constant) by setting $k=r$ and $m=0$. We now use this
Lemma to make statements about the density of neighborhoods.


\begin{lemma}\label{lem:lcycles}
    The probability that there is an $r$-neighborhood
    in $\gndn$ with $m$ more edges than vertices
    is at most $f(r,m)d(n)^{2r}(d(n)^{2r+1}/n)^m$
    for some function~$f$.
\end{lemma}
\begin{proof}
    Consider any $r$-neighborhood with $\ell$ vertices. Assume the
    neighborhood contains at least $m$ more edges than vertices. Let $T$ be a
    breadth-first search spanning tree of this neighborhood. Since $T$ contains
    $\ell$ vertices and $\ell-1$ edges, there are $m+1$ edges which are not
    contained in $T$. Each extra edge is incident to two vertices. Let
    $U$ be the set of these vertices.
    Let $H$ be the graph induced by the union of the
    $m+1$ extra edges and the unique paths in $T$ from $u$
    to the root of $T$ for each $u\in U$.
    Since $|U| \le 2(m+1)$ and each path to the root in
    the breadth-first-search tree $T$ has length at most $r$, the number of
    vertices of $H$ is bounded by $2r(m+1)$.

    In summary, if there exists an $r$-neighborhood with at least $m$ more
    edges than vertices then there exists a subgraph with $k \le 2r(m+1)$
    vertices and $m$ more edges than vertices. But according to
    Lemma~\ref{lem:num_h}, the expected number of such subgraphs is bounded by
    $$
    \frac{\Big(\big(2r(m+1)\big)^2d(n)\Big)^{2r(m+1)+m}}{n^{m}} = f(r,m)d(n)^{2r}
    \Bigl(\frac{d(n)^{2r+1}}{n}\Bigr)^m.
    $$
    This also bounds the probability that such a subgraph exists.
    \qed
\end{proof}
%

\begin{theorem}
    Let $d(n) = n^{o(1)}$.  Then $\gndn$ has \aas locally bounded treewidth.
\end{theorem}
\begin{proof}
    The show that a graph has locally bounded treewidth we have to
    show that the treewidth of every $r$-neighborhood is bounded by a
    function of~$r$ alone.

    Since $d(n) = n^{o(1)}$, there exists a monotone decreasing function
    $g(n)$ with $d(n) \le n^{g(n)}$ and $\lim_{n\to\infty} g(n) = 0$. Let $h(r)$
    be the inverse function of $1/8g(r)$. Since $g(n)$ is monotone
    decreasing, $h(r)$ exists and is monotone increasing. We show that for all
    $r\ge0$ every subgraph with radius $r$ has \aas treewidth at most $h(r)$.
    We distinguish between two cases. The first case is $r < 1/8g(n)$
    and $f(r,1)<n^{1/4}$.

    According to Lemma~\ref{lem:lcycles}, an $r$-neighborhood of $G$ has more
    edges than vertices with probability at most
    $$
    f(r,1)\frac{d(n)^{4r+1}}{n} \le f(r,1)n^{g(n) (4 \frac{1}{8g(n)}+1)-1} \le
    f(r,1)n^{-1/2 + g(n)} = o(1)
    $$
    We can conclude that every $r$-neighborhood has \aas treewidth at
    most~$2$.

    The second case is $r\geq1/8g(n)$, which means $h(r) \ge n$, so even the
    treewidth of the whole graph is \aas bounded by~$h(r)$ and the third case
    is given by $f(r,1)\geq n^{1/4}$ and the (total) treewidth is \aas bounded
    by $f(r,1)^4$. 

    Altogether, the treewidth of an $r$-neighborhood is \aas bounded by 2, by
    $h(r)$, or by $f(r,1)^4$. \qed
\end{proof}

\section{Algorithm for Subgraph Isomorphism}
\label{sec:algorithms}
In this section we solve \textsc{Subgraph Isomorphism}, which given a graph
$G$ and a graph $H$ asks, whether $G$ contains $H$ as a subgraph. This is
equivalent to FO-model checking restricted to only existential quantifiers.

%
Let $H$ be a connected graph with $h$ vertices and radius $r$. In this section
we discuss how fast it can be decided whether $\gndn$ contains $H$ as a
subgraph. We first discuss the runtime of simple branching algorithms on
\ErRe graphs and how exploiting local structure may lead to better run-times.
We discovered that if the radius $r$ of the pattern graph is small, an
approach based on local structure is significantly faster.

For low-degree graphs there exists a simple branching algorithm to decide
whether a graph $G$ contains $H$ as a subgraph in time $O(\Delta^h n)$, where
$\Delta$ is the maximal degree in $G$. Let us first assume that $d(n)=d$ is
constant. There is nevertheless a non-vanishing probability that the maximal
degree of $\gnd$ is as large as $\sqrt{\log n}$. Therefore, the maximal degree
cannot be bounded by any function of $d$. This implies that a naive, maximal
degree based algorithm may have at least a quasi-linear dependence on $n$,
while we present an algorithm which has only a linear dependence on $n$.

Let us also assume that that $d(n)$ is of order $\log{n}$ and even that the
maximum degree is bounded by $O(d(n))$. A naive branching algorithm may
therefore decide whether $\gndn$ contains $H$ in expected time $O(d(n))^h n$.
We improve this result, not making any assumption about the maximal degree, by
replacing the factor $O(d(n))^h$ in the runtime with
$2^{O(h)}(d(n)\log n)^{O(r)}$,
where $r$ is the radius of~$H$. For graphs with small radius, the runtime is
no longer dominated by a factor $O(d(n))^h$.  The new algorithm
may be significantly smaller when $d(n)$ is, for example, of order $\log n$.

So far we only discussed connected subgraphs. Using color-coding techniques,
the results in this section can easily be extended to disconnected subgraphs,
where the radius of each component is bounded by $r$. Color-coding may,
however, lead to an additional factor of $c^h$ in the runtime: Assume $H$ has
$c$ components where the size of $H$ is $h$. We want to color each vertex of
$G$ uniformly at random. Assume $G$ contains $H$, then the probability that
every component of $H$ can be embedded using vertices of a single color is at
least $1/c^h$. So if $H$ can be embedded in $G$ we will answer yes after an
expected number of $c^h$ runs.

For the following result notice that if $d(n)$ is poly-logarithmic in
$n$ the runtime is quasi-linear in $n$. For $d(n)=n^{o(1)}$ the
dependence on $n$ is $n^{1+o(1)}$. The algorithm is given in the proof
for Theorem~\ref{lem:findH}.

\begin{lemma}\label{lem:nhood-small}
    In $\gndn$ holds with probability of at least $1-n^{-\frac14\log(n)}$
    that every $r$-neighborhood has size at most $\log(n)^{2r} d(n)^{r}$.
\end{lemma}
\begin{proof}
    The Chernoff Bound states for the degree $D$ of an individual vertex
    that $\P[D \ge  x] \le e^{-(\frac13\frac{x}{d(n)}-1)d(n)}$
    and therefore
    $\P[D \ge  \log(n)^{2} d(n)] \le n^{-\frac14\log(n)}$.
    Let $\hat D$ be the maximal degree of the graph.
    With the union bound we have a similar bound for $\hat D$.
    Every $r$-neighborhood has size at most $\hat D^r$.
\end{proof}

\begin{theorem}\label{lem:findH}
    Let $H$ be a connected graph with $h$ vertices and radius $r$.

    There is a deterministic algorithm that can find out whether $H$
    occurs as a subgraph in $\gndn$ in expected time 
    $2^{O(h)}(d(n)\log n)^{O(r)}n$.
\end{theorem}
\begin{proof}
    We sketch the algorithm briefly.  The algorithm works on
    a graph $G=\gndn$.  In the following we assume that
    every $r$-neighborhood in $G$ has size at most $d(n)^r\log(n)^{2r}$.
    By Lemma~\ref{lem:nhood-small} this assumption holds with a probability
    of at least $1-n^{\log(n)/4}$ and we can easily check it within
    the stated time bounds.
    Should the assumption be wrong, we can use a brute force
    algorithm without affecting the average running time.

    In a preprocessing step we look at the connected graph~$H$ and
    construct a subgraph~$H'$ that is also connected, but consists
    only of a tree with two additional edges (if possible, otherwise
    we set $H'=H$).
    
    We enumerate all $r$-neighborhoods in~$G$ and try to find $H$ in
    every one of them as follows:  By using color-coding we
    enumerate all subgraphs in the $r$-neighborhood that are
    isomorphic to~$H'$.  This can be done by using the algorithm for 
    finding a graph of bounded
    treewidth~\cite{Alon:1995:COL:210332.210337}
    with the enumeration techniques in~\cite{Chen:2006:EEN:2162624.2162644}.
    The expected time needed is
    $2^{O(h)}(d(n)\log(n))^{O(r)}$ times the number of subgraphs
    that are found.  However, by Lemma~\ref{lem:num_h}
    the latter number is bounded by a constant.

    After enumerating all subgraphs isomorphic to $H'$ we have to find
    out whether $G$ contains $H$ as a subgraph.  If this turns out to
    be true, then $H$ can be found only somewhere where $H'$ was
    found.  Hence, it suffices to look at all found $H'$ in $G$ and see
    whether by adding a subset of the possible $\binom{h}{2}$ edges we
    can find~$H$.  This can be done in time
    $O(2^{h^2}d(n)^r\log(n)^2)$, which is asymptotically faster than
    the remaining part.
\end{proof}

\section{Experimental Evaluation of \BaAl-Graphs}
\label{sec:experiments}
In the previous section, we showed that \ErReGrs have bounded expansion for
edge probability $p=d/n$ (with constant $d$) and are nowhere dense with $p = n^{o(1)}/n$. In this
section, we discuss the sparsity of the \BaAl model. It is known that this
model has \emph{not} \aas bounded expansion, because it contains an unbounded
clique with non-vanishing probability~\cite{StrucSpars}. It is not known,
however, if it is (or is not) \aas somewhere-dense. Our experiments seem to
imply that on average \BaAlGrs seem to be dense but that this density is
limited to early vertices: In the \BaAl model, vertices with high degree tend
to be preferred for new connections. This means that edge probabilities are
not independent. Moreover, the expected degree $d(i) = \sqrt{n/i}$ for a
vertex $i$ is less uniform than it is for \ErReGrs, where
$d(i)=pn$.

To evaluate the expansion properties of the \BaAl-model, we compute transitive
fraternal augmentations and $p$-centered colorings. These have been introduced
by \Nesetril and Ossona de Mendez, and are highly related to bounded expansion
and a tool for developing new and faster algorithms. A graph class has bounded
expansion if and only if the maximum in-degree of transitive fraternal
augmentations is bounded, or the graph admits a $p$-centered coloring with
bounded number of colors.

\begin{definition}[Transitive fraternal augmentation \cite{NOdM08}]
Let $\overrightarrow{G}$ be a directed graph. A \textit{1-transitive fraternal
augmentation} of $\overrightarrow{G}$ is a directed graph $\overrightarrow{H}$
with the same vertex set, including all the arcs of $\overrightarrow{G}$ and
such that, for any vertices $x$, $y$, $z$,
\begin{itemize}
\item if ($x$, $z$) and ($z$, $y$) are arcs of $\overrightarrow{G}$ then
($x$, $y$) is an arc of $\overrightarrow{H}$ (\textup{transitivity}),
\item if ($x$, $z$) and ($y$, $z$) are arcs of $\overrightarrow{G}$ then
($x$, $y$) or ($y$, $x$) is an arc of $\overrightarrow{H}$ (\textup{fraternity}).
\end{itemize}
A \textit{transitive fraternal augmentation} of a directed graph
$\overrightarrow{G}$ is a sequence $\overrightarrow{G}_1
\subseteq \dots \subseteq \overrightarrow{G}_i \subseteq
\dots \subseteq \overrightarrow{G}_n$, such that $\overrightarrow{G}_{i+1}$ is
a 1-transitive fraternal augmentation of $\overrightarrow{G}_{i}$.
\end{definition}

\begin{definition}[$p$-centered coloring \cite{nesetril:08-1}]
\label{def:p-center}
For an integer $p$, a $p$-\textit{centered coloring} of $G$ is a coloring of
the vertices such that any connected subgraph $H$ induced on the vertices of
an arbitrary set of $i$ colors ($i \leq p$), $H$ must have at least one color
that appears exactly once.
\end{definition}
Showing that the maximum in-degree of a transitive fraternal augmentation or
the number of colors needed for a $p$-centered coloring does not grow with the
size of the graph is a way to prove that a graph has bounded
expansion~\cite{NOdM08}. When designing algorithms, $p$-centered colorings can
be used to solve hard problems efficiently. By using $p$-centered colorings,
we can decompose a graph into small, well-structured subgraphs such that
$\npclass$-hard problems can be solved easily on each subgraph before
combining these small solutions to get a solution for the entire graph. It is
important that the number of colors needed for a $p$-centered coloring for a
fixed $p$ is small, as the runtime usually is a function of the number of
colors needed. If a graph class does not have bounded expansion; that is, the
number of colors grows with $n$, but very slowly, such as $\log{\log{n}}$,
using these algorithms might still be practical.

One example problem which can be solved directly using $p$-centered colorings
is \textsc{Subgraph Isomorphism}, where one asks if a graph $H$ is contained
in a graph $G$ as a subgraph. In general graphs, this problem is
$\wclass$-hard when parameterizing by the size of
$H$~\cite{parametrized_complexity}. However, there exist an algorithm, whose
runtime is a function of the number of colors needed for a $p$-centered
coloring, where $p$ depends on the size of $H$~\cite{sparsity}. So, regardless
of the fact whether \BaAlGrs are theoretically sparse or not, calculating the
number of colors of a $p$-centered coloring for different graph sizes has
direct impact on the feasibility of a whole class of algorithms on these
graphs.

\subsection{Experiment Overview} \label{sec:pipeline}

We analyze the expansion properties of \BaAlGrs by computing transitive
fraternal augmentations and $p$-centered colorings. In the following, we
describe the heuristics used to compute these. In order to compute the
transitive fraternal augmentations of a graph $G$, the graph is oriented to a
directed graph $\overrightarrow{G}_1$ by using low-degree orientation, in
which every edge $(u, v)$ in $G$ is transformed to an arc $(u, v)$ in
$\overrightarrow{G}_1$ if the degree of $u$ is greater than the degree of $v$.
Then, transitive fraternal augmentations are applied to
$\overrightarrow{G}_1$, which yield a sequence $\overrightarrow{G}_1,
\overrightarrow{G}_2, \dots, \overrightarrow{G}_i$. The augmentation heuristic
we used was proposed in earlier work~\cite{felix:15}: To build graph
$\overrightarrow{G}_i$ from $\overrightarrow{G}_{i-1}$ we need to perform
transitive fraternal augmentations. First we create the set $F$ of fraternal
edges of $\overrightarrow{G}_{i-1}$. Let $G_F$ be the graph induced by $F$.
Now we can orient the edges of $G_F$ by the same low-degree orientation
performed earlier to get the directed fraternal edges $\overrightarrow{F}$
that are added to $\overrightarrow{G}_{i-1}$ and result in
$\overrightarrow{G}_{i}$. Now we can color the undirected graph $G_i$ of
$\overrightarrow{G}_i$ by iterating through the vertices in a descending-
degree order and assign each vertex the lowest color that does not appear in
its neighborhood. We then check whether that coloring is a $p$-centered
coloring of the input graph $G$. If this is not the case we repeat this
procedure for $G_{i+1}$.

\begin{algorithm}
\caption{Computing $p$-centered colorings}
\label{alg:pcc}
\begin{algorithmic}[1]
\Procedure{Compute\textendash $p$-centered\textendash Colorings}{$G, p$}
\State Create $\overrightarrow{G}_1$ from $G$ using low-degree orientation
\State $i \gets 1$
\Loop
\State $c \gets $Greedy coloring of the undirected graph $G_i$
\If{$c$ is a $p$-centered coloring of $G$}
\State \Return{$c$}
\EndIf
\State Compute 1-transitive fraternal augmentation of $\overrightarrow{G}_i$ to get $\overrightarrow{G}_{i+1}$
\State $i \gets i+1$
\EndLoop
\EndProcedure
\end{algorithmic}
\end{algorithm}

\subsection{\BaAl Graphs are Empirically Dense}

First, we analyze the maximum in-degree of transitive fraternal augmentations.
We ran the previously described algorithm on random \BaAlGrs with $d = 2$ for
different sizes ($500 \leq n \leq 3000$) and calculated the maximum in-degree
of up to five transitive fraternal augmentation steps. The results are shown
in Figure~\ref{fig:result1a}. Each data point is an average over ten runs with
the same $n$. For all graphs both the maximum in-degree grows with $n$, which
would not be the case for graphs with bounded expansion.

To evaluate how well the expansion properties of \BaAlGrs can be practically
exploited, we analyzed the number of colors needed to construct $p$-centered
colorings. We constructed $3$-~and $4$-centered colorings. with the same graph
parameters and sizes than before. The results are shown in
Figure~\ref{fig:result1b}. For the analyzed range, the number of colors needed
grows steadily. Furthermore, the number of colors needed to construct
$4$-centered colorings is substantially higher than the number of colors
needed for $3$-centered colorings. Computing higher order colorings or
colorings for larger graphs was infeasible with the used algorithm. It seems
practically impossible to use $p$-centered colorings algorithmically for
\BaAlGrs. We have to note that the used algorithm is only a heuristic and the
real values might be much better than what we have computed. But since these
heuristics work well for graphs that have low treedepth colorings, it is
unlikely that the graphs have bounded coloring number for $p$-centered
colorings.
\begin{figure}
\centering
\subfloat[Transitive fraternal augmentations]{
\includegraphics[width=0.4\textwidth]{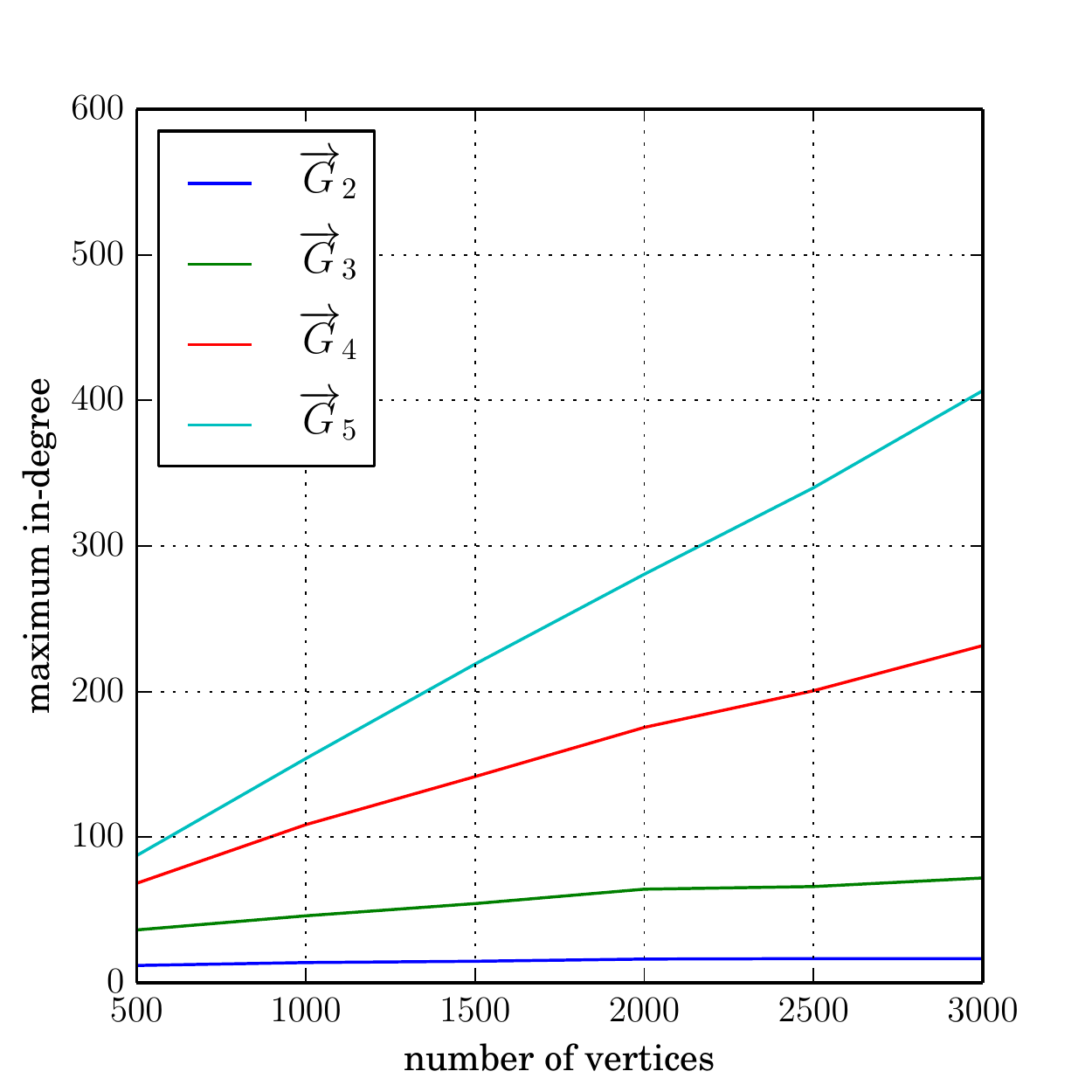}
\label{fig:result1a}
}%
\subfloat[$p$-centered colorings]{
\includegraphics[width=0.4\textwidth]{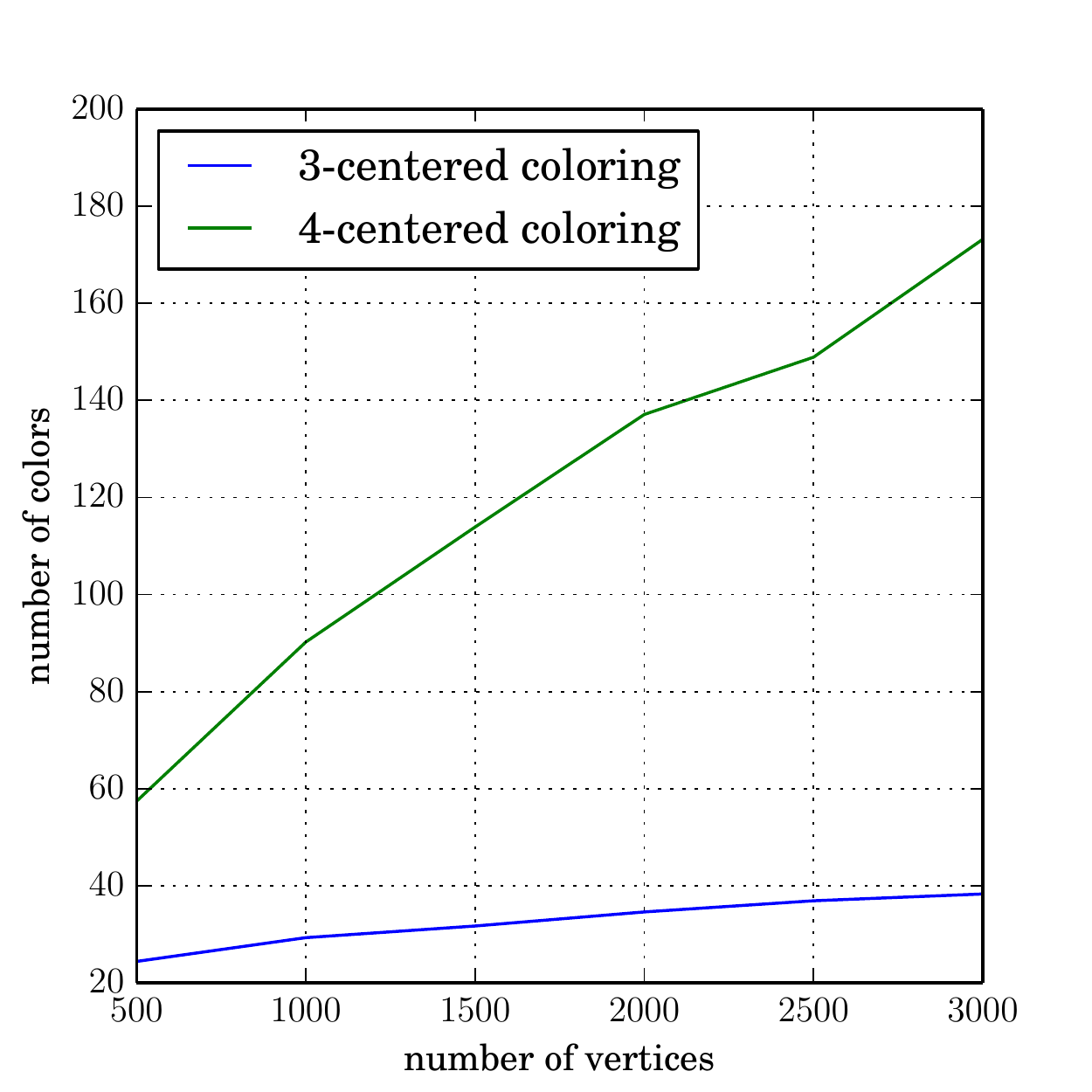}
\label{fig:result1b}
}%
\caption{Results for \BaAlGrs with $d = 2$ for increasing $n$.}
\label{fig:result1}
\end{figure}
\subsection{Density Seems Limited to Early Vertices}
\begin{figure}
\centering
\subfloat[Transitive fraternal augmentations]{
\includegraphics[width=0.4\textwidth]{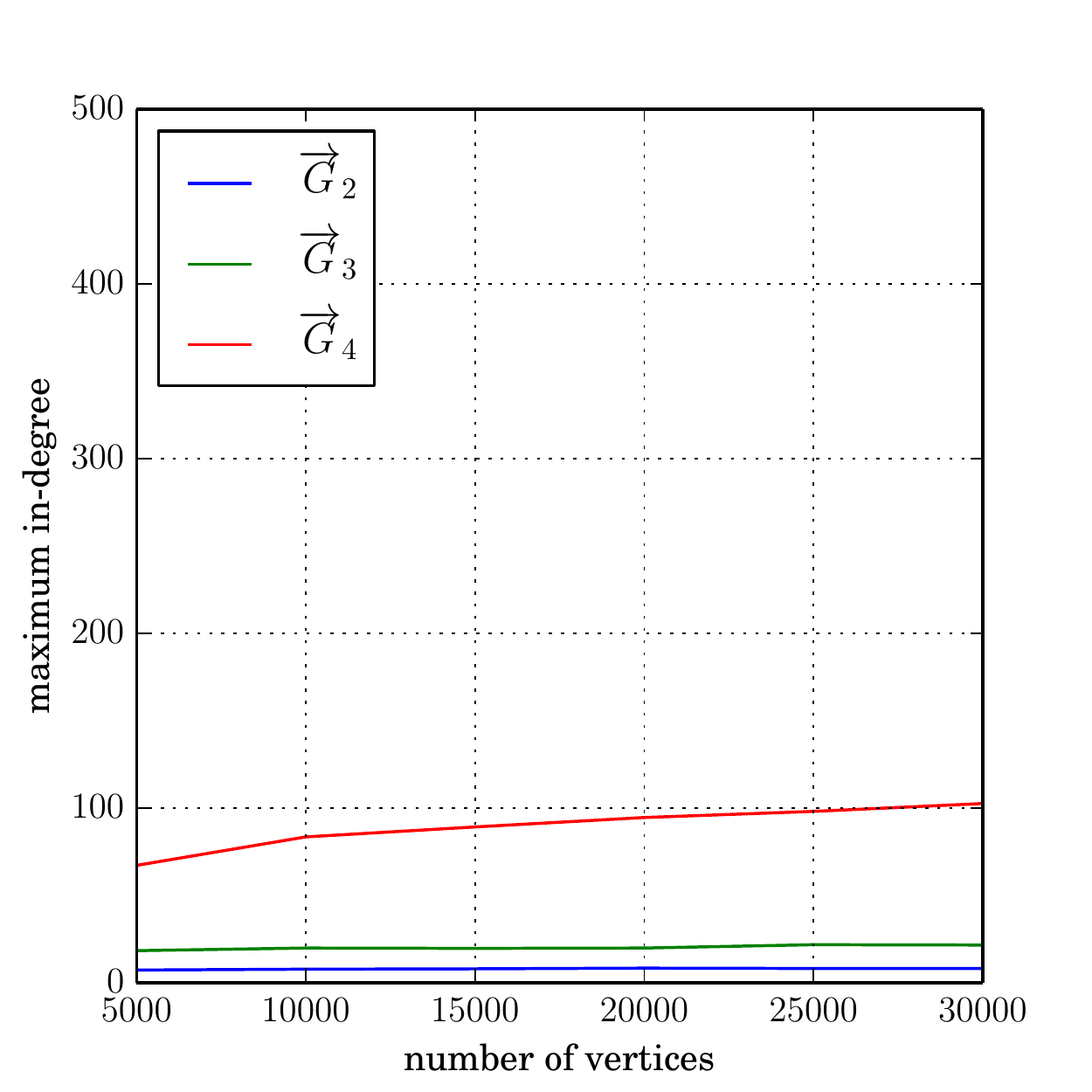}
}
\subfloat[$p$-centered colorings]{
\includegraphics[width=0.4\textwidth]{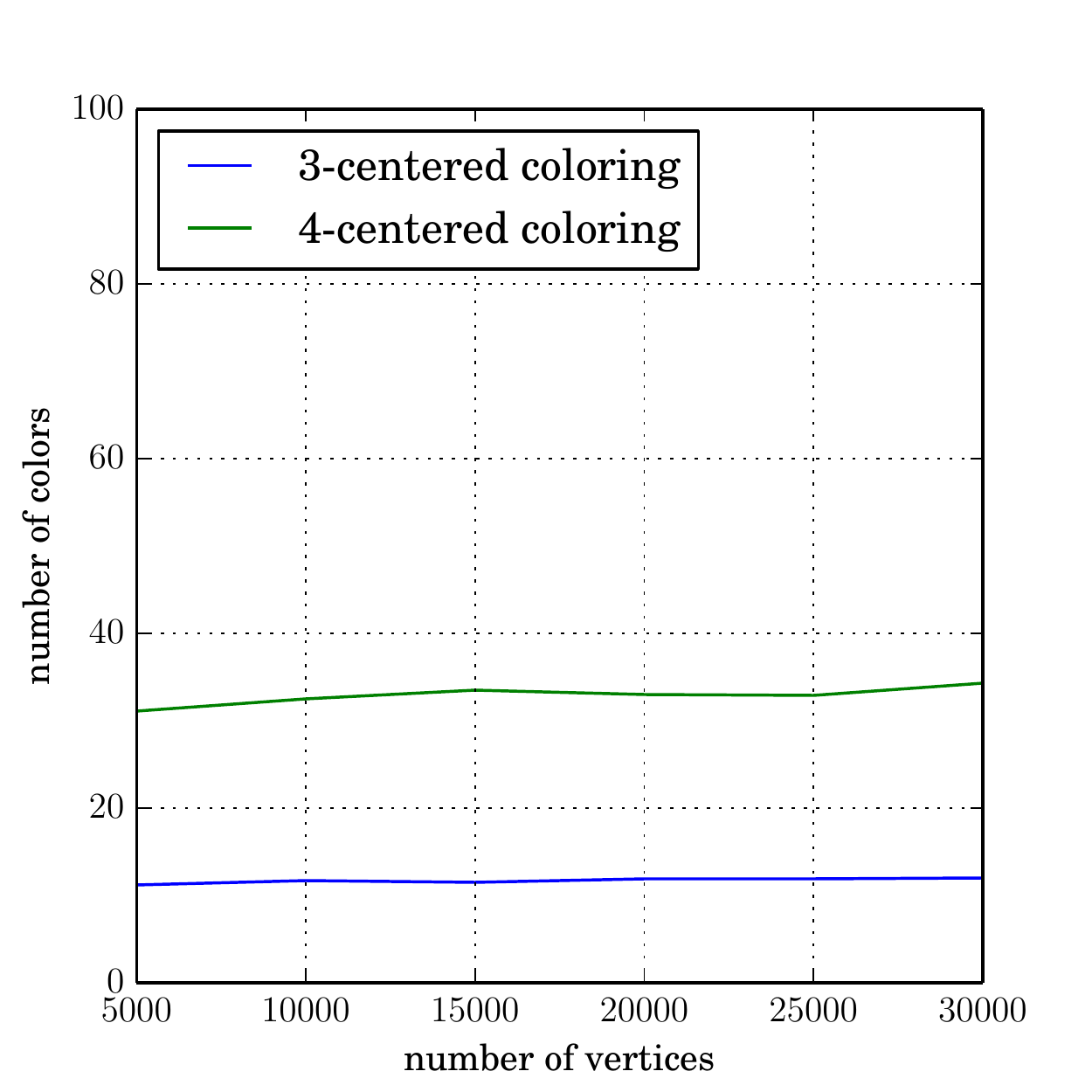}
}
\caption{Results for \BaAlGrs with $d = 2$ for increasing $n$ after deleting
the first 10\% of vertices.}
\label{fig:result2}
\end{figure}
Previously, we showed that the colors needed to construct $p$-centered
colorings of small graphs can be very high. In this section we discover that
the early vertices of the random process heavily affect these results. We
remove the first 10\% of the vertices added in the random process and analyze
the maximum in-degree of transitive fraternal augmentations and number of
colors needed to construct $p$-centered colorings. By removing those 10\%, we
can construct $p$-centered colorings for much larger graphs ($5000 \leq n \leq
30000$), see Fig.~\ref{fig:result2}. The required number of colors for
$p$-centered colorings and maximum in-degree of transitive fraternal
augmentations remain stable and do not seem to depend on the number of
vertices. This suggests that these 10\% of the early vertices contain almost
all of the density of \BaAlGrs. This is of course a linear factor and it
remains to see if one can use much smaller functions of $n$, like for example
$\log{n}$. The sizes of the graphs at hand, however, were not large enough to
investigate sub-linear functions of $n$ with a meaningful result.

\section{Conclusion}

In this work we gave an alternative proof that $\gnd$ has \aas bounded
expansion and have shown that $\gndn$ with $d(n) = n^{o(1)}$ has \aas locally
bounded treewidth. Our results are based on the fact that local neighborhoods
of \ErReGrs are tree-like with high probability. It is
known~\cite{Grohe07logic} that for a graph $G = \gndn$ with $d(n)=n^{o(1)}$
and a first-order formula $\varphi$ one can decide whether $G \models \varphi$
in expected time $f(|\varphi|)n^{1+o(1)}$ for some functions $f$ and $g$. This
result can also be proven using our techniques. It remains to show whether it
is possible to answer this question in linear expected fpt-time (where $d(n)n$
is the expected number of edges), \ie $O(f(|\varphi|)d(n)n)$. In this paper,
we also presented a more efficient algorithm for the subgraph isomorphism
problem on \ErReGrs if the pattern graph has small radius. It would be
interesting to consider other measures for the pattern graph as well, such as
treewidth or treedepth. Furthermore, we gathered empirical evidence which
suggests that \BaAlGrs are somewhere dense. It would be interesting to prove
this conjecture.

\bibliographystyle{splncs}
\bibliography{references}
\end{document}